\documentclass[lettersize,journal]{IEEEtran}
\usepackage{amsmath,amsfonts,amssymb}
\usepackage{algorithmic}
\usepackage{algorithm}
\usepackage{array}
\usepackage{textcomp}
\usepackage{stfloats}
\usepackage{url}
\usepackage{verbatim}
\usepackage{graphicx}
\usepackage{cite}
\usepackage{comment}
\usepackage{physics}
\usepackage{epstopdf}
\usepackage{wrapfig}
\usepackage{xcolor}
\usepackage{subcaption}
\usepackage{romannum}
\usepackage{titlesec}
\usepackage{hyperref}
\titlespacing*{\section}
{0pt}    
{1ex}    
{0.5ex}  

\titlespacing*{\subsection}
{0pt}
{0.8ex}
{0.4ex}
\newtheorem{thm}{\textbf{Theorem}}
\newtheorem{lemma}{\textbf{Lemma}} 
\newtheorem{remark}{\textbf{Remark}}
\newtheorem{assumption}{\textbf{Assumption}}
\newtheorem{definition}{\textbf{Definition}}
\newtheorem{Proposition}{\textbf{Proposition}}
\newtheorem{cor}{\textbf{Corollary}}
\newtheorem{note}{\textbf{Note}}
\newenvironment{proof}{{\emph{\textbf{Proof:}} }}{\hfill $\square$}
\newtheorem{problemx}{\textbf{Problem}}
\newenvironment{problem}{%
  \problemx
}{\endproblemx}
\begin{document}

\title{Bearing Only Distributed Circumnavigation with Limited Target Information for Asymmetric Dubins Vehicles}

\author{Kushal P. Singh$^{1}$ Student Member IEEE, and Twinkle Tripathy$^{2}$ Member IEEE 
\thanks{*This work was not supported by any organisation.}
\thanks{$^{1}$Kushal P. Singh is a research scholar in the department of Electrical Engineering, Indian Institute of Technology Kanpur,
        Kanpur-208016, India
        {\tt\small kushalp20@iitk.ac.in}}%

\thanks{$^{2}$Twinkle Tripathy is an assistant professor with the Department of Electrical Engineering, Indian Institute of Technology Kanpur,
        Kanpur-208016, India
        {\tt\small ttripathy@iitk.ac.in}}%
}

\markboth{Journal of \LaTeX\ Class Files,~Vol.~14, No.~8, August~2021}%
{Shell \MakeLowercase{\textit{et al.}}: A Sample Article Using IEEEtran.cls for IEEE Journals}

\IEEEpubid{0000--0000/00\$00.00~\copyright~2021 IEEE}

\maketitle

\begin{abstract}
In this paper, we present a class of bearing-based distributive nonlinear guidance laws for the cooperative circumnavigation of a stationary target by a heterogeneous team of asymmetric Dubins' vehicles. In such a vehicle, the maximal left and right turn capabilities are non-uniform. In the given framework, the target's location is known only to a small subset of the vehicles, called the \textit{leaders}. The uninformed vehicles, called the \textit{followers}, use information from their out-neighbours in the communication graph, constructed using the nearest-neighbour rule. A class of guidance laws is formulated that relies solely on the heading angle and line-of-sight (LOS) angles of a vehicle's designated out-neighbour in the graph. Using Zubov’s theorem, we prove that the proposed guidance laws achieve global asymptotic stability (GAS) under angular-speed-only control and ensure the convergence of the trajectories of all the Dubins vehicles to a common centre. The proposed results are validated through numerical simulations. 
\end{abstract}


\section{Introduction}
\label{Section:Introduction}
Target circumnavigation has become important in many autonomous systems applications~\cite{B.Jha,li2025bearing} including surveillance, reconnaissance, search and rescue, agricultural robotics, orbit maintenance, multi-robot firefighting, and oil and gas pipeline inspection. Its practical significance has motivated extensive research on the design of effective guidance laws.

A substantial body of this research is dedicated to developing guidance laws for a single pursuer circumnavigating a stationary target~\cite{2022_single,hashemi2015unmanned,milutinovic2014coordinate,milutinovic2017coordinate}. More specifically, a substantial amount of it is based on range information-based guidance law design, whereas comparatively less attention has been given to guidance laws that rely on bearing information~\cite{deghat2014localization}. The control law presented in \cite{2022_single} allows an agent to achieve a desired standoff distance from a target by using only the rate of change of range information, thereby, eliminating the need for absolute position data. For a Dubins vehicle, an approach to achieve circumnavigation under a noisy environment using both range rate measurements and bearing angle is presented in \cite{hashemi2015unmanned}. Thereafter, using information similar to that in \cite{2022_single}, the authors in \cite{milutinovic2014coordinate,milutinovic2017coordinate} presented a control law for a Dubins vehicle that utilises feedback from range and range rate measurements. On the other hand, the authors in~\cite{deghat2014localization} achieve the target's circumnavigation solely from bearing measurements.

Although single-agent systems offer considerable utility, the inherent complexity of large-scale networks, as observed in biological systems, demands a transition to cooperative control methodologies. The benefits of cooperative strategies, including scalability, robustness against individual failures, and mitigation of communication bottlenecks, have positioned them as a central focus of intensive research. Seminal works in this area, including \cite{all_to_all_communication,limited_communication,marshall2004formations}, have formulated control laws for groups of unicycle pursuers operating at identical constant linear speeds. A Lyapunov-based control strategy for achieving cyclic formations with all-to-all communication was introduced in \cite{all_to_all_communication} and was later extended in \cite{limited_communication} to support limited communication topologies. Inspired by the bug algorithm, a cyclic pursuit-based control law was proposed in \cite{marshall2004formations} to encircle a target in regular polygon formations, a concept that was subsequently generalised in \cite{sinha2007generalization} to handle pursuers with varying linear speeds, thereby, achieving rotating formations. For constant linear speeds, the authors in \cite{TRIPATHY2024111315} propose to achieve convergence from almost all initial conditions.

For agents modelled as single and double integrators, distributed control laws are developed in \cite{ramirez2010distributed} to achieve circular, elliptic, and spiral formations using information from only two pursuers. In \cite{el2012distributed}, a distributed control law is designed in which the centre's location is dynamic and contingent on the system's initial conditions. The authors of \cite{seyboth2014collective} introduced a guidance law for a heterogeneous group of unicycles, wherein two distinct types of circular motion are achievable: one defined by common angular speeds and another by common radii. Further investigating control heterogeneity, the authors of \cite{zheng2015distributed} modulate both linear and angular speeds to attain circular formations at different radii by assigning attractive or repellent behaviours to agents relative to the target and their neighbours.

A significant challenge in cooperative circumnavigation is when not all agents know the target's location. The works presented in \cite{yu2018circular} and \cite{yu2018distributed} tackle this issue by regulating both the linear and the angular speeds to achieve cooperative formation when only a single unicycle is aware of the target's location. Yet another challenge pertains to the kinematics of the agents.
In all the works presented for Dubin's vehicles in \cite{2022_single,hashemi2015unmanned,milutinovic2014coordinate,milutinovic2017coordinate,dong2019circumnavigation}, symmetric bounds are assumed on angular speed in both clockwise and anticlockwise directions, which is not always realistic. Asymmetric control limits can arise in several practical scenarios, such as an aircraft with a damaged aileron or missing wingtip \cite{bakolas2009asymmetric}, systems influenced by gyroscopic effects from a rotating engine \cite{lundgren2020asymmetric}, and marine vehicles subject to external disturbances like water currents \cite{parlangeli2024novel}. Hence, there is a need for asymmetric Dubins vehicles.

Unlike most of the above-mentioned works that consider the same radii for all the pursuers, configurations with pursuers at varying radii provide substantial strategic benefits, such as enabling inner pursuers to conduct surveillance on a target while outer pursuers form a protective cordon against potential threats. To achieve such multi-layered formations, in this paper, we introduce a method that operates cooperatively within a multi-agent framework, subject to constraints on the communication graph and control inputs. Unlike many existing studies that rely on distance measurements or dual linear-angular speed control for stability, this paper addresses circumnavigation without distance measurements, guaranteeing stability by controlling only the angular speeds of asymmetric Dubins vehicles, thereby facilitating real-world implementation. The main contributions of this paper are:
\begin{enumerate}
    \item \textit{Novel class of bearing-based distributive guidance laws for asymmetric Dubins vehicles:} We introduce a class of bearing-based distributive guidance laws that utilises only angular speed control to circumnavigate a stationary target by a heterogeneous group of $n$ asymmetric Dubins vehicles. To the best of our knowledge, this framework is the first to cooperatively solve the circumnavigation problem while accounting for both different constant linear speeds and unequal turning rate limits, a more realistic model for physical systems.
    \item \textit{Global asymptotic stability (GAS):} The proposed guidance laws ensure circumnavigation by the entire group of Dubins vehicles from any arbitrary initial conditions. GAS is established using \textit{Zubov's theorem} \cite{khalil2002nonlinear}.
    \item \textit{Robust with limited information and input saturation:} The guidance law for followers operates on minimal, local information (a single neighbour's heading angle and LOS angle) without requiring target data. We prove that convergence is guaranteed even under asymmetric input saturation, demonstrating the strategy's robustness to real-world constraints.
\end{enumerate}

The structure of this paper is outlined as follows. Section \ref{Section:preliminaries} introduces the essential background and preliminary concepts. In Section \ref{Section:Problem_Formulation}, the main problem is formally defined. The proposed guidance laws are detailed in Section \ref{section:Guidance law}, first for a group of unicycles (subsection \ref{subsec:for_two}) and then for a group of asymmetric Dubins vehicles (subsection \ref{subsec:dubin}). Simulation results supporting the theoretical developments are provided in Section \ref{Sec:Simulaton}. Lastly, Section \ref{Sec:Conclusion} offers concluding remarks and outlines possible directions for future work.
\section{Preliminaries}
\label{Section:preliminaries}
\textit{Notations:} Let $\mathbb{R}$ and $\mathbb{N}$ denote the set of real numbers and natural numbers, respectively. $\mathbb{R}^+$ denotes the set of positive real numbers and $\mathbb{R}^n$ denotes a real valued vector of size $n \in \mathbb{N}$. $||\bullet||$ denotes the two-norm of a vector. $\iota$ denotes the imaginary number.

\textit{Graph theory:} A directed graph (digraph) $\mathcal{G}=(\mathcal{V},\mathcal{E})$, where $\mathcal{V}=\{1,2,\dots,n\}$ is the set of nodes (pursuers) and $\mathcal{E}\subseteq \mathcal{V}\times \mathcal{V}$ is the set of edges representing communication links.
A directed edge $(i,j)\in\mathcal{E}$ from $i$ to $j$ implies that pursuer $i$ accesses information from pursuer $j$. Then, $j$ becomes an out-neighbour of $i$, and $i$ becomes an in-neighbour of $j$. The set of all the out-neighbours of $i$ are denoted by $\mathcal{N}^{out}_i$.In this framework, the edge $(i,j)\in\mathcal{E}$ implies that pursuer $i$ can sense the heading angle of pursuer $j$ and LOS angle from $i$ to $j$.
A node with no outgoing edges is called a \textit{sink node}.
A digraph $\mathcal{G}$ is called strongly connected if there exists a directed path from any node to every other node.
A digraph $\mathcal{G}$ is called weakly connected if the undirected version of $\mathcal{G}$ is connected.

\textit{Nonlinear systems:}  A nonlinear system is represented as $\dot{x}=f(t,x,u)$, where $f:D \to \mathbb{R}^n$ and $D$ is the domain. We refer to $x$ as the state, $t$ as time and $u$ as the input. It is referred to as an autonomous system if $\dot{x}=f(x)$. The equilibrium point ($x=0$) of $\dot{x}=f(x)$ is:
\begin{itemize}
    \item Stable if for each $\epsilon>0$, there exist $\delta=\delta(\epsilon)>0$ such that $||x(0)||<\delta  \implies  ||x(t)||< \epsilon$, $\forall$ $t \geqslant0$.
    \item Asymptotically stable if it is stable and $\delta$ an be chosen such that $||x(0)||<\delta \implies \lim_{t \to \infty}x(t)=0$.
    \item Globally asymptotically stable if it is asymptotically stable and domain $D=\mathbb{R}^n$.
\end{itemize}
\begin{lemma}{(Zubov's Theorem)}
    \label{thm:zubov}
    Consider the non-linear system $\dot{x}=f(x)$ with an equilibrium at origin and let $G \subset \mathbb{R}^n$ be a domain containing the origin. Suppose there exist two functions $V:G\rightarrow \mathbb{R}$ and $h:\mathbb{R}^n \rightarrow \mathbb{R}$ with the following properties: 
    \begin{itemize}
        \item $V$ is continuously differentiable and positive definite in $G$ and satisfies
        \setlength{\abovedisplayskip}{2pt}
        \setlength{\belowdisplayskip}{2pt}
        \begin{equation}
            \label{eq:cond_V}
            0<V(x)<1, \quad \forall x\in G \setminus\{0\}.
        \end{equation}
        \item As $x$ approaches the boundary of $G$, or in case of unbounded $G$ as $||x||\rightarrow\infty$, $\lim V(x)=1$.
        \item $h$ is continuous and positive definite on $\mathbb{R}^n$.
        \item For $x\in G$, $V(x)$ satisfies the partial differential eqn.
        \begin{equation}
        \label{eq:zubov_condition}
            \frac{\partial V}{\partial x}f(x)=-h(x)[1-V(x)].
        \end{equation}
    \end{itemize}
    \begin{figure}[ht]
\begin{center}
\includegraphics[scale=0.85]{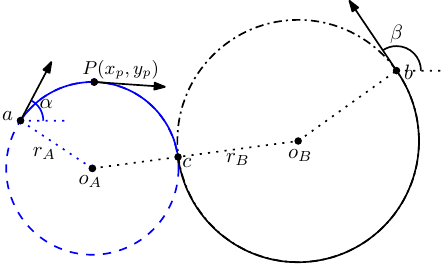}    
\caption{Location of centres $\mathbf{o_A}$ and $\mathbf{o_B}$ of circles $\mathcal{C}_A$ and $\mathcal{C}_B$, respectively}  
\label{fig:centers and radii}                                 
\end{center}                                 
\end{figure}
    Then, $x=0$ is asymptotically stable, and $G$ is the region of attraction.
\end{lemma}

\textit{Path traversal for an autonomous vehicle:} Consider the traversal problem of an autonomous vehicle starting at pre-specified oriented points $A(\mathbf{a},\alpha)$ and $B(\mathbf{b},\beta)$, respectively. Suppose the traversal is along the two circles $\mathcal{C}_A$ $\&$  $\mathcal{C}_B$ of radii $r_A$ $\&$ $r_B$, which are tangent to each other as depicted in Fig. \ref{fig:centers and radii}. The coordinates of the centres $\mathbf{o_A}$ and $\mathbf{o_B}$ are given by $\mathbf{o_A}=(a_x+ r_{A}\sin{\alpha}, a_y - r_{A}\cos{\alpha})$ and $\mathbf{o_B}=(-r_{B}\sin{\beta}, r_{B}\cos{\beta})$, respectively. The specification of such a circle-circle (CC) trajectory is given below:
\begin{lemma}[\cite{rao2024curvature}]
\label{lem:feas_traj}
For any values of oriented points $A(\mathbf{a},\alpha)$ and $B(\mathbf{b},\beta)$, a unique CC trajectory exists $\forall r_A\in\mathbb{R}$ such that $r_B$ is given by $r_B=p_3/(r_A-p_1)+p_2$ where:
\begin{align*}  
    p_1&=(a_x\sin\beta-a_y\cos\beta)/(1-\cos(\alpha-\beta)),\\
    p_2&=(a_x\sin\alpha-a_y\cos \alpha)/(1-\cos(\alpha-\beta)),\\
    p_3&=\left(\dfrac{a_x\sin((\alpha+\beta)/2)-a_y\cos((\alpha+\beta)/2))}{1-\cos(\alpha-\beta)}\right)^2.
\end{align*}
\end{lemma}
\section{Problem formulation}
\label{Section:Problem_Formulation}
Circumnavigation, in which pursuers maintain circular trajectories around a central target (like a landmark, beacon etc.), is pivotal for many operational scenarios. These range from surveillance, reconnaissance, and environmental sensing to coordinated multi-robot firefighting. Driven by these practical requirements, we address this problem for a group of $n$ heterogeneous autonomous pursuers.
\begin{figure}[ht]
\begin{center}
\includegraphics[scale=1.5]{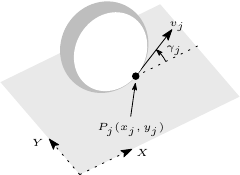}    
\caption{Unicycle model}  
\label{figure:unicycle_model}                                 
\end{center}                                 
\end{figure}

The pursuers are modelled as simplified point mass unicycles. The kinematics of a unicycle $j\in\mathcal{V}$ is described by:
\setlength{\abovedisplayskip}{2pt}
\setlength{\belowdisplayskip}{2pt}
\begin{equation}
\label{eq:unicycle_kinematics}
\dot{x_j} = v_j \cos \gamma_j, \quad
\dot{y_j} = v_j \sin \gamma_j, \quad
\dot{\gamma}_j = u_j,
\end{equation}
Here, $P_j(x_j,y_j)\in \mathbb{R}^2$ represents the position coordinates of pursuer $j$, $v_j\in \mathbb{R}^+$ is its constant linear speed, $\gamma_j \in \mathbb{S}^1$ is its heading angle relative to a fixed reference frame as shown in Fig. \ref{figure:unicycle_model}, and $u_j$ is the angular control input.
 
To better approximate real-world constraints where data availability is limited, we assume target location data is available only to a specific subset of the group. We define them as a partition of two disjoint sets: leaders ($\mathcal{L}$), who have access to the target coordinates, and followers ($\mathcal{F}$), who function without this direct information. While the unicycle model is inherently under-actuated, we impose an additional constraint to enhance practical feasibility: the forward speed $v_j$ is held constant, and control is exerted solely through asymmetric bounded angular input $u_j$. While this simplifies practical implementation, it exacerbates the control theoretical challenge by further reducing the system's actuation capabilities. 

That said, our system comprises two types of pursuers: leaders and followers. Since leaders possess more information, their control strategy should be less challenging than that of the followers. We therefore leverage an existing result from our previous work in Lemma \ref{lem:feas_traj}, in which we establish that a vehicle can transition between any two oriented points using CC trajectories. As it grants leaders the requisite degrees of freedom to independently prescribe $\alpha$, $\beta$, and $r_B$ (see Fig. \ref{fig:centers and radii}), thereby facilitating collision avoidance. This implies that any initial state can be guided to a desired circular orbit of a specified radius around the target. 

In scenarios involving multiple leaders, if the leaders can sense one another, it is not necessary to guide each leader independently for target circumnavigation. Instead, the leaders can be coordinated through a distributed guidance strategy, which improves scalability, robustness, operational efficiency, and resource utilisation, while also eliminating the issue of a `single point of failure'. To ensure such a distributed implementation, the following rule is imposed for leaders under mutual sensing.

\begin{Proposition}
\label{propo:leaders}
Consider a set of at least two leaders $\mathcal{L} = \{l_1, l_2, \dots\}$ such that each leaders can sense few other leaders, which can either be zero or non-zero. Suppose one leader $l_k \in \mathcal{L}$ is assigned the task of circumnavigating the target at a desired radius $R_{l_k}$. Then, the remaining leaders determine their respective desired radii according to the following rule:
\begin{subequations}
\label{eq:rule}
\begin{align}
\label{eq:rule_part1}
R_{l_i} &= \frac{r_{l_i}(0)}{\beta} + \alpha\, i
&& \text{if }  r_{l_i}(0) \geqslant r_{\mathcal{N}^{l_i}_{out}}(0), \\
\label{eq:rule_part2}
R_{l_i} &= \frac{r_{l_i}(0)}{\beta}
&& \text{if $r_{l_i}(0) < r_{\mathcal{N}^{l_i}_{out}}(0)$ or $r_{\mathcal{N}^{l_i}_{out}} \in \phi$},
\end{align}
\end{subequations}
where $\beta \in \mathbb{R}^+$ and $\alpha \in \mathbb{R}$ are any constants as long as $R_{l_i}>0$, $\phi$ denotes an empty set, $r_{l_i}(0)$ is the distance of each leader $l_i$ from the target at time $t=0$, and $\mathcal{N}^{l_i}_{out}$ denotes the set of out-neighbour leaders of \(l_i\). Accordingly, every leader follows Lemma~\ref{lem:feas_traj} to achieve their desired radii.
\end{Proposition}

This approach effectively guides the leaders to circumnavigate the target at prescribed radii. Now, our remaining primary objective is to achieve the same outcome for our followers. So, we write the problem statement formally.

\begin{problem}
\label{prob:main_problem}
Consider a group of heterogeneous asymmetric Dubins' vehicles governed by the kinematics given in eqn.~\eqref{eq:unicycle_kinematics}, operating under a distributed information framework in which only the leaders know the target location. The objective is to design a distributed guidance law \(u_j\), for every \(j \in \mathcal{V}\), such that each vehicle asymptotically converges to and circumnavigates a prescribed circular orbit around the target.
More precisely, for any initial condition $(x_j(0),\, y_j(0),\, \gamma_j(0)) \in \mathbb{R}^2 \times \mathbb{S}^1$,
the closed-loop system must ensure that $\lim_{t \to \infty} \left| r_j(t) - R_j \right| = 0,
\quad \forall\, j \in \mathcal{V}$,
where \(r_j(t)\) denotes the instantaneous distance of vehicle \(j\) from the target, and \(R_j > 0\) denotes the desired constant circumnavigation radius assigned to vehicle \(j\).
\end{problem}

With the problem statement formally defined, we present our main results in the following section.
\section{Main results}
\label{section:Guidance law}
In this section, we formulate guidance laws to direct followers into circumnavigating the target, ensuring their angular speeds synchronise with those of their designated leaders. We do so in two steps; we begin by analysing the fundamental scenario of a single unicycle $i$ circumnavigating its out-neighbour $j$, which serves as the basis for multi-agent convergence. We then generalise this to a group of $n$ pursuers. Finally, we extend the methodology to accommodate the complexities of asymmetric Dubins vehicles.

To formalise the interaction structure among the pursuers, we introduce the following graph theoretic definitions.
\begin{definition}[Sensing graph]
    \label{def:sensing_graph}
    The sensing graph is defined as $\mathcal{G}_S=(\mathcal{V},\mathcal{E}_S)$, where $i,j \in \mathcal{V}$ denote the nodes corresponding to the pursuers, and a directed edge $(i,j)\in \mathcal{E}_S$ exists if and only if pursuer $i$ can sense pursuer $j$.
\end{definition}
\begin{figure}[ht]
    \centering
        \begin{subfigure}[b]{0.2\textwidth}
        \centering    
        \includegraphics[width=\linewidth]{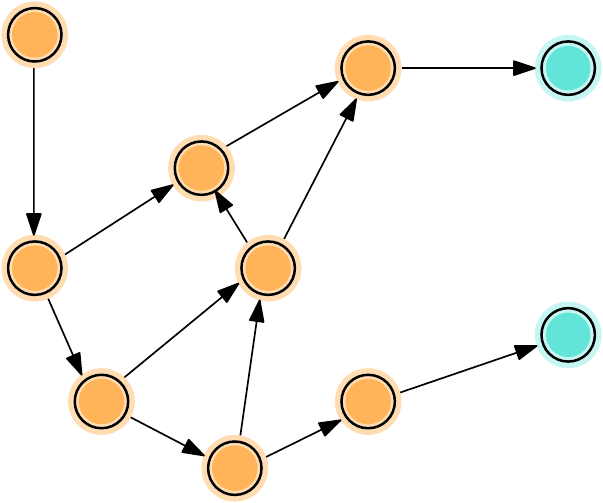}
        \caption{Sensing graph $\mathcal{G}_S$.}
        \label{fig:sensing_graph}
    \end{subfigure}
    \hfill
    \begin{subfigure}[b]{0.24\textwidth}
        \centering
        \includegraphics[width=0.9\linewidth]{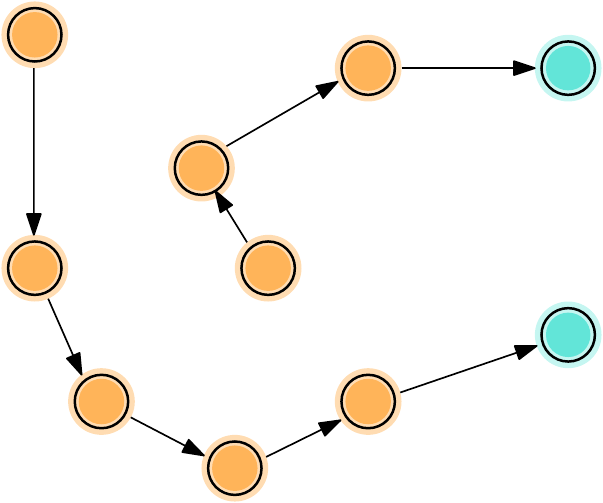}  
        \caption{Communication graph $\mathcal{G}_C$.}
        \label{fig:sense_to_communication_graph}
    \end{subfigure}
    \caption{Leaders and followers are shown in blue and orange colours, respectively}
    \label{fig:to_communication}
\end{figure}

\begin{definition}
    \label{def:communii_graph}
    The communication graph $\mathcal{G}_C=(\mathcal{V},\mathcal{E}_C)$ is a subgraph of $\mathcal{G}_S$. In this graph, each pursuer $i$ selects as its out-neighbour the closest pursuer it can sense in $\mathcal{G}_S$. If multiple nearest candidates exist, the pursuer may arbitrarily select any one of them (see Fig. \ref{fig:to_communication}). Consequently, each follower node has exactly one out-neighbour. Additionally, the leaders are sink nodes.
\end{definition}

Hence, $\mathcal{G}_C$ is not unique by construction and may contain multiple disconnected components even for a strongly connected $\mathcal{G_S}$. Furthermore, we assume the following about $\mathcal{G}_C$.
\begin{assumption}
    \label{assump:infor_flow}
    In the communication graph $\mathcal{G}_C$, we assume that every follower $f\in \mathcal{F}$ has a directed path to at least one leader in the set $\mathcal{L}$.
\end{assumption}

As discussed earlier, the pursuers in the leader set $\mathcal{L}$ have information about the target location and can guide the followers toward it. Assumption \ref{assump:infor_flow} ensures that the information from the leaders propagates to every follower in $\mathcal{F}$. With this, we proceed for the design of the guidance law.

\subsection{Convergence for a group of unicycles}
\label{subsec:for_two}
\begin{figure}[ht]
    \centering
    \includegraphics[scale=0.8]{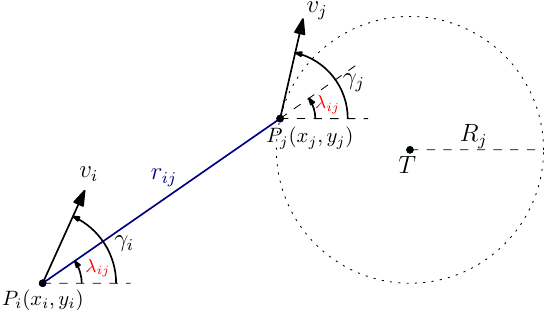}
    \caption{Engagement geometry}
    \label{fig:engage_geom}
\end{figure}
Consider a pursuer $i$, situated at a distance $r_{ij}$ from a pursuer $j$, which has already converged to a circular orbit around the target using Lemma \ref{lem:feas_traj}, as illustrated in Fig. \ref{fig:engage_geom}. Without loss of generality (WLOG), we assume the target is located at the origin $O$. The positions of pursuers $i$ and $j$ are denoted by $P_i$ and $P_j$, respectively. The LOS angle of the line connecting $i$ and $j$ at $P_i$ is denoted by $\lambda_{ij}$. The kinematics in polar coordinates are given by:
\begin{subequations}
\label{eq:kinematics_polar}
\begin{align}
\label{eq:(r_ij_dot)}
 \dot{r}_{ij} & = v_j\cos{(\gamma_j-\lambda_{ij})}-v_i\cos{(\gamma_i-\lambda_{ij})}  \\
 \label{eq:(lambda_ji_dot)}
 r_{ij}\dot{\lambda}_{ij} &= v_j\sin{(\gamma_j-\lambda_{ij})}-v_i\sin{(\gamma_i-\lambda_{ij})}.
 \end{align}
\end{subequations}

For any pursuer $i$ to orbit the same centre as its out-neighbours $j$ (already on a circular path) without range information, following geometrical conditions must be satisfied.
\begin{lemma}
    \label{lemm:unique_equilibrium}
    Consider any two pursuers $i,j\in \mathcal{V}$ with unicycle kinematics given in eqn. \ref{eq:unicycle_kinematics}, such that $(i,j)\in \mathcal{E}_C$. Suppose the pursuer $j$ traverses a fixed circular trajectory of radius $R_j$ (centred at target $T$ with angular speed $\omega_j$, see Fig.~\ref{fig:engage_geom}). Then the pursuer $i$ is guaranteed to move on a concentric circle with angular speed $\omega_j$ under the following two conditions.
\begin{enumerate}
    \item[a)] the heading angle difference remains constant, i.e. $\dot{\gamma}_j(t)-\dot{\gamma}_i(t)=0$,
    \item[b)] the offset angle between the velocity vector and LOS remains constant, i.e. $\dot{\lambda}_{ij}(t)-\dot{\gamma}_i(t)=0$.
\end{enumerate}
\end{lemma}
\begin{proof}
The position of pursuer $j$ can be expressed as $\overrightarrow{P_j}(t)=x_j(t)+ \iota y_j(t)$. Further, since pursuer $j$ is circumnavigating the target, its position vector can be written as $\overrightarrow{P_j}(t)=R_j \exp{(\iota\theta_j(t))}$ where the radius $R_j$ is constant and $\theta_j$ is the angular coordinate. With reference to Fig. \ref{fig:engage_geom}, it follows from the circular geometry that $\gamma_j(t)=\theta_j(t)-\pi/2$. Its velocity vector is then given by $\overrightarrow{v_j}(t)=\frac{d}{dt}\overrightarrow{P_j}(t)=\iota R_j\dot{\theta}_j(t)\exp{(\iota\theta_j(t))}$. A constant $v_j$ necessitates that the $\dot{\theta}_j(t)$ is also constant.

Now we analyse the relative motion between the two pursuers under conditions (a) and (b) of the lemma statement, and we suitably define two error variables for pursuer $i$:
    \begin{subequations}
    \label{eq:error variables}
    \begin{align}
    \label{eq:e_1^j}
     e_{i_1} &\triangleq \gamma_j - \gamma_i \\
    \label{eq:e_2^j}
     e_{i_2} &\triangleq \lambda_{ij} - \gamma_i
    \end{align}
    \end{subequations}

First, we analyse the implication of condition (a) of the lemma statement. We start by substituting the condition $\gamma_i(t) = \gamma_j(t) - e_{i_1}$ into the expression for $\overrightarrow{v_i}(t)$ and simplifying, we get $\overrightarrow{v_i}(t)=(v_i/v_j)(\exp{(-\iota e_{i_1}}))\overrightarrow{v_j}(t)$.

So, $\overrightarrow{v_i}(t)$ is proportional to $\overrightarrow{v_j}(t)$ by a constant, $\alpha \triangleq (v_i/v_j)\exp{(-\iota e_{i_1})}$. Hence, exploiting the velocity and position relationship, we get $\overrightarrow{P_i}(t)=\alpha \overrightarrow{P_j}(t)+c$, where $c$ is a constant of integration. Further, $\overrightarrow{P_i}(t)=\alpha R_j\exp{\iota\theta_j(t)}+c$ as $\overrightarrow{P_j}(t)=R_j \exp{(\iota\theta_j(t))}$. Note that the equation describes a circular path for pursuer $i$ of radius $|\alpha|R_j$ and centred at the point $c$.
Next, on applying condition (b) of the lemma statement, the angular separation $e_{i_2}$ (as defined in eqn. \eqref{eq:e_2^j}) is constant. The LOS vector from $i$ to $j$ is $\overrightarrow{r_{ij}}(t) = \overrightarrow{P_j}(t)-\overrightarrow{P_i}(t)$, which is given by:
\begin{equation}
\label{eq:integrated_P_j}
\overrightarrow{r_{ij}}(t) = \overrightarrow{P_j}(t) - (\alpha\overrightarrow{P_j}(t)+c) = (1-\alpha)\overrightarrow{P_j}(t)-c.
\end{equation}
Expressed as a function of $\theta_j$, eqn. \eqref{eq:integrated_P_j} is $\overrightarrow{r_{ij}}(\theta_j) = (1-\alpha)R_j\exp{(\iota\theta_j(t))}-c$.
Then, condition (b), combined with $\gamma_j(t)=\theta_j(t)-\pi/2$, implies that $\arg{(\overrightarrow{r_{ij}}(\theta_j))} - \theta_j$ must be a constant. For this relationship to hold, its derivative w.r.t. $\theta_j$ must be zero:
\begin{equation}
\label{eq:lead_angle_condition}
    \dfrac{d}{d\theta_j}(\arg{\overrightarrow{r_{ij}}(\theta_j(t))} - \theta_j(t)) = 0.
\end{equation}
Using the identity for the derivative of an argument, $\frac{d}{d\theta_j}(\arg{f(\theta_j)})=\text{Img}\left(f'(\theta_j)/f(\theta_j)\right)$:
\begin{equation}
\label{eq:zero_condition}
\text{Img}\{(\overrightarrow{r_{ij}}'(\theta_j))/(\overrightarrow{r_{ij}}(\theta_j))\} - 1 = 0.
\end{equation}
We compute the derivative $\overrightarrow{r_{ij}}'(\theta_j)$ and the ratio becomes:
\begin{equation}
    \label{eq:dash_by_normal}
    \dfrac{\overrightarrow{R_{ij}}'(\theta_j)}{\overrightarrow{R_{ij}}(\theta_j)}=\dfrac{\iota(1-\alpha)R_j\exp{(\iota\theta_j})}{(1-\alpha)R_j\exp{(\iota\theta_j)}-c}.
\end{equation}
Substituting the result of the eqn. \eqref{eq:dash_by_normal} in the condition from the eqn. \eqref{eq:zero_condition}, we get:
\begin{equation}
\label{eq:arg_img}
 \text{Img}\left(\dfrac{\iota(1-\alpha)R_j\exp{(\iota\theta_j)}}{(1-\alpha)R_j\exp{(\iota\theta_j)}-c}\right) -1 = 0.  
\end{equation}
 If we test $c=0$, the expression simplifies to $\text{Img}(\iota) - 1 = 1-1 = 0$, which is true. If $c \neq 0$, the term in the eqn. \eqref{eq:arg_img} does not simplify to zero. Therefore, the only possible solution that satisfies the condition for all time is $c=0$.

Since $c=0$, the trajectory for $i$ is $\overrightarrow{P_i}(t)=\alpha\overrightarrow{P_j}(t)$. This confirms that pursuer $i$ circumnavigates the same centre $O$ as pursuer $j$ and shares the same angular velocity $\omega_j$.
\end{proof}

\begin{remark}
    \label{remark:rigid_body}
    When Lemma \ref{lemm:unique_equilibrium} holds, the triangle $\triangle TP_iP_j$ is a rigid body. Its side lengths,the radii $R_j = \|P_j(t) - O\|$ and $R_i = \|P_i(t) - O\|$, and the inter-agent distance $r_{ij} = \|P_j(t) - P_i(t)\|$ are all invariant for all $ t \in \mathbb{R} $.
\end{remark}

Lemma \ref{lemm:unique_equilibrium} identifies constant values of $e_{1i}$ and $e_{2i}$ as the necessary conditions for circumnavigation around a common centre. To implement these conditions across a multi-agent system without range or range-rate information, we present guidance laws in the following theorem.

\begin{thm}
\label{thm:two_pursuers}
    Consider a system of $n$ pursuers governed by kinematics defined in eqn. \eqref{eq:unicycle_kinematics}. Let every leader $l \in \mathcal{L}$ move on a stable circular path such that $\lim_{t\to t_f} r_l(t) = \bar{r}_l$ and $\dot\gamma_l(t)\to \bar{\omega}_l$. Under Assumption \ref{assump:infor_flow}, each follower $i\in \mathcal{F}$ is guaranteed to circumnavigate the target $T$ at the same angular rate as the leader at the terminus of its directed path in $\mathcal{G}_C$ under the following distributed guidance law:
    \begin{equation}
    \label{eq:follower_guidance}
    \dot{\gamma}_i = g(e_{i_1},e_{i_2}).
    \end{equation}
    Here, the control function $g(\cdot)$ must be chosen such that the function $h(\mathbf{e}_i-\bar{\mathbf{e}}_{i}) \triangleq -2\alpha (\mathbf{e}_i-\bar{\mathbf{e}}_{i})^T f(\mathbf{e}_i)$ remains continuous and positive definite in $\mathbb{R}^2 \setminus \{\mathbf{0}\}$, where function $f(\mathbf{e}_i)=\dot{\mathbf{e}}_i$. In this context, $\mathbf{e}_i=[e_{i_1} \quad e_{i_2}]^T$ represents the error states and $\bar{\mathbf{e}}_{i}=[\bar{e}_{i_1} \quad \bar{e}_{i_2}]$ denotes error states at equilibrium.
\end{thm}
\begin{proof} 
We proceed in two steps. The first step validates the convergence of a single follower whose immediate out-neighbour in $\mathcal{G}_C$ is circumnavigating the target. The second step generalises this result to the rest of the system by applying inductive reasoning over the communication graph.

 Consider a single follower $i \in \mathcal{F}$ having an out-neighbour $j \in \mathcal{L}$ that is already on a stable circular path (see Fig. \ref{fig:engage_geom}). Now based on the equilibrium conditions from Lemma \ref{lemm:unique_equilibrium}, the two error variables $e_{i_1} = \gamma_j - \gamma_i$ and $e_{i_2} = \lambda_{ij} - \gamma_i$ are constant at equilibrium. Let the state vector be $\mathbf{e}_i = [e_{i_1} \quad e_{i_2}]^T$. The system dynamics are given by:
\begin{equation}
\label{eq:m_dot}
\dot{\mathbf{e}}_i = f(\mathbf{e}_i) = [\dot{e}_{i_1} \quad \dot{e}_{i_2}]^T = [f_1 \quad f_2]^T.
\end{equation}

Substituting the values of $\dot{e}_{i_1}$ and $\dot{e}_{i_2}$ in eqn. \eqref{eq:m_dot} gives: 
\begin{equation}
\label{eq:f(e)}
f(\mathbf{e}_i) = \begin{bmatrix} 
\omega_j - g(e_{i_1},e_{i_2}) \\
\dfrac{v_j \sin(e_{i_1} - e_{i_2}) + v_i \sin(e_{i_2})}{r_{ij}} - g(e_{i_1},e_{i_2})
\end{bmatrix}
\end{equation}

Let the equilibrium value of $\mathbf{e}_i$ be $\bar{\mathbf{e}}_{i}$. The equilibrium values $\bar{e}_{i_1}$ and $\bar{e}_{i_2}$ in eqn. \eqref{eq:f(e)} gives the condition:
\begin{equation}
    \label{eq:eq_cond}
    \omega_j = g(\bar{e}_{i_1},\bar{e}_{i_2}).
\end{equation}

With the equilibrium defined, we proceed to analyse its stability properties. We employ Zubov's Theorem for this purpose, as it allows the simultaneous characterisation of the region of attraction and the stability of the origin. Accordingly, we define a new state variable $\mathbf{z}_i=\mathbf{e}_i-\bar{\mathbf{e}}_{i}$ to shift equilibrium to origin, where $\mathbf{z}_i=[z_{i_1} \quad z_{i_2}]^T$. The new system dynamics are: 
\begin{equation}
    \label{eq:new_variable}
    \dot{\mathbf{z}}_i\triangleq\tilde{f}(\mathbf{z}_i)=\dot{\mathbf{e}}_i=f(\mathbf{z}_i+\bar{\mathbf{e}}_{i}).
\end{equation} 

Substituting $e_{i_1}=z_{i_1}+\bar{e}_{i_1}$, $e_{i_2}=z_{i_2}+e_{i_2}$, and $\omega_j$ (from eqn. \eqref{eq:eq_cond}) into eqn. \eqref{eq:f(e)}, we obtain:
\begin{equation}
\label{eq:modified_f}
\tilde{f}(\mathbf{z}_i)
=
\begin{bmatrix}
\tilde{f}_{i_1} \quad \tilde{f}_{i_2}
\end{bmatrix}^T,
\end{equation}
where $\tilde{f}_{i_1}= g(\bar{e}_{i_1},\bar{e}_{i_2})-g(z_{i_1}+\bar{e}_{i_1},z_{i_2}+\bar{e}_{i_2})$ and $\tilde{f}_{i_2}=
 (v_j \sin(z_{i_1}-z_{i_2}+\bar{e}_{i_1}-\bar{e}_{i_2})
 + v_i \sin(z_{i_2}+\bar{e}_{i_2})
 )/r_{ij}-g(z_{i_1}+\bar{e}_{i_1},z_{i_2}+\bar{e}_{i_2})$.

We choose functions $V(\mathbf{z}_i)=1-\exp{(-\alpha((z_{i_1})^2+(z_{i_2})^2))}$, where $\alpha \in \mathbb{R}^+$, and $h(\mathbf{z}_i)=-2\alpha[z_{i_1} \quad z_{i_2}]\tilde{f}(\mathbf{z}_i)$. These functions satisfy the conditions laid down in Lemma \ref{thm:zubov} in Section \ref{Section:preliminaries}. The function $V$ is continuously differentiable, positive definite in $\mathbb{R}^2$ ($G=\mathbb{R}^2$), and satisfies eqn. \eqref{eq:cond_V}. Furthermore, the limit of $V$ as $||\mathbf{z}_i||\rightarrow \infty$ is $1$.

The function $g(e_{i_1},e_{i_2})$ is selected to ensure that $h$ is positive definite. We then verify the condition:
\begin{equation}
    \label{eq:zubov_proof}
    \frac{\partial V}{\partial \mathbf{z}_i}\tilde{f}(\mathbf{z}_i)=-h(\mathbf{z}_i)[1-V(\mathbf{z}_i)].
\end{equation}

For this purpose, we substitute the partial derivative of $V(\mathbf{z}_i)$ and the system dynamics $\tilde{f}(\mathbf{z}_i)$ into the left-hand side of eqn. \eqref{eq:zubov_proof}. We obtain:
\begin{align} 
 \frac{\partial V}{\partial \mathbf{z}_i}\tilde{f}(\mathbf{z}_i)&= \underbrace{2\alpha e^{-\alpha ((z_{i_i})^2)+(z_{i_2})^2)} [z_{i_1} \quad z_{i_2}]}_{\partial{V}/\partial{\mathbf{z}_i}}\underbrace{[\tilde{f}_{i_1} \quad \tilde{f}_{i_2}]^T }_{\tilde{f}(\mathbf{z}_i)}  \notag \\  &= 2\alpha e^{-\alpha ((z_{i_1})^2)+(z_{i_2})^2)} \left( z_{i_1} \tilde{f}_{i_1} + z_{i_2} \tilde{f}_{i_2} \right). \label{eq:derivation_step2} 
\end{align}

Recalling proposed $h(\mathbf{z}_i)$ from above and observing that $1-V(\mathbf{z}_i) = e^{-\alpha ((z_{i_1})^2)+(z_{i_2})^2)^2}$, we can rewrite eqn. \eqref{eq:derivation_step2} as $-h(\mathbf{z}_i) [1-V(\mathbf{z}_i)]$, which verifies eqn. \eqref{eq:zubov_proof}.

As per Lemma \ref{thm:zubov} stated in Section \ref{Section:preliminaries}, the origin is asymptotically stable. Additionally, the region of attraction is $\mathbb{R}^2$. Hence, the origin is globally asymptotically stable. Next, we prove the stability of the entire group.

We define the radial error $d_{i_r}(t) \triangleq r_i(t)-\bar{r}_i$, where $r_i(t)$ is the distance from the target and $\bar{r}_i$ is a constant value. Leaders $l \in \mathcal{L}$ converge by definition in theorem statement, so $\lim_{t\to t_f} d_{l_r}(t) = 0$ and $\dot\gamma_l \to \bar{\omega}_l$.

Now, consider a path $l_m \leftarrow f_1 \leftarrow \dots \leftarrow f_k$. Assume follower $f_k$ has converged to a stable circular path with speed $\bar{\omega}_{l_m}$. For the next follower $f_{k+1}$, the pursuer $f_k$ acts as a stable leader. Applying the result from stage one of the proof, $f_{k+1}$ must also converge to a circular path and match the angular speed of $f_k$ (and thus of $l_m$). 
Since every follower has a path to a leader (Assumption \ref{assump:infor_flow}), this inductive argument cascades to all $f \in \mathcal{F}$. Hence, all pursuers converge to their respective radii and angular speeds. Hence, proved.
\end{proof}

Having established that a populous system can converge even with limited target information. We now address a practical constraint: achieving the same with input saturation.  
\subsection{Convergence for a group of asymmetric Dubins vehicles}
\label{subsec:dubin}
In this subsection, we present convergence results for a group of asymmetric Dubins vehicles. A Dubins vehicle is modelled as a unicycle moving with a constant linear speed and a bounded angular speed. Its kinematics are:
\begin{equation}
    \dot{x}_d = v_d \cos \gamma_d, \qquad
    \dot{y}_d = v_d \sin \gamma_d, \qquad
    \dot{\gamma}_d = u_d,
    \label{eq:kinematics_dubin}
\end{equation}
where $u_d \in [-\omega_{\max},\, \omega_{\max}]$ denotes the bounded turning rate. This model is consistent with the unicycle framework considered in this work, provided the control input $u_d$ is constrained within prescribed bounds. Such constraints are physically motivated, accounting for factors such as actuator saturation and computational limitations.

In general, the admissible control set can be expressed as $u_d = [u_{\min},\, u_{\max}]$, where typically $u_{\min} = -\omega_{\max}$ and $u_{\max} = \omega_{\max}$. However, the assumption of symmetric bounds, i.e., $|u_{\min}| = |u_{\max}|$ may not always be realistic. Asymmetric control limits can arise in several practical scenarios, such as an aircraft with a damaged aileron or missing wingtip \cite{bakolas2009asymmetric}, systems influenced by gyroscopic effects from a rotating engine \cite{lundgren2020asymmetric}, and marine vehicles subject to external disturbances like water currents \cite{parlangeli2024novel}. Such vehicles are called asymmetric Dubins vehicles. In such cases, it is more appropriate to model the different turning rates in both directions. Therefore, we present our results with different turning rates. But before presenting the convergence result for the entire system under input bounds, we analyse the behaviour of the desired control input from eqn. \eqref{eq:follower_guidance} when subject to saturation.
\begin{thm}
\label{lem:stauration_behaviour}
Consider a pursuer $j$ moving on a circular path with a constant angular speed $\omega_j$. A follower $i \in \mathcal{F}$ whose desired (unsaturated) angular rate is given by eqn. \eqref{eq:follower_guidance} in Theorem \ref{thm:two_pursuers}.

Assume the applied control is saturated such that 
$\dot{\gamma}_i(t)\in[\dot{\gamma}_i^{\mathrm{lower}},\,\dot{\gamma}_i^{\mathrm{upper}}]$ 
with $\dot{\gamma}_i^{\mathrm{lower}} < 0 < \dot{\gamma}_i^{\mathrm{upper}}$, and the actuator limits satisfy 
$\dot{\gamma}_i^{\mathrm{lower}} < \omega_j < \dot{\gamma}_i^{\mathrm{upper}}$.

Then, if the actuator is driven to either of its saturation limits, the desired control $u_i(t)$ necessarily exits saturation in finite time. Specifically:
\begin{enumerate}
  \item If $\dot{\gamma}_i(t) = \dot{\gamma}_i^{\mathrm{upper}}$ for all $t \geqslant t_0$, then there exists $t_1>t_0$ such that $u_i(t_1) < \dot{\gamma}_i^{\mathrm{upper}}$.
  \item If $\dot{\gamma}_i(t) = \dot{\gamma}_i^{\mathrm{lower}}$ for all $t \geqslant t_0$, then there exists $t_2>t_0$ such that $u_i(t_2) > \dot{\gamma}_i^{\mathrm{lower}}$.
\end{enumerate}
Consequently, saturation cannot persist indefinitely, and the desired control is always driven back into the admissible range.
\end{thm}
\begin{proof}
Differentiating errors $e_{i_1}$ and $e_{i_2}$ gives $\dot{e}_{i_1}(t)=\dot{\gamma}_j(t)-\dot{\gamma}_i(t)=\omega_j-\dot{\gamma}_i(t)$. Since pursuer $j$
is on a circular path with constant angular speed, $\dot{\gamma}_j(t)\equiv\omega_j$. Now, we examine the two saturation cases:

\textit{Case A (Upper Saturation):} Assume $\dot{\gamma}_i(t)=\dot{\gamma}_i^{\mathrm{upper}}$ for all $t\geqslant t_0$. Then, $\dot{e}_{i_1}=\omega_j-\dot{\gamma}_i^{upper}<0$, hence $e_{i_1}(t)$ tends to $-\infty$ as $t$ tends to $\infty$. The desired (unsaturated) control input $u_i(t)=g(e_{i_1}(t),e_{i_2}(t))$ contains a term that is strictly decreasing w.r.t. $e_{i_1}$ and the term $e_{i_2}$ is finite. Consequently, $u_i(t)$ tends to $-\infty$ as $t$ tends to $\infty$. Therefore, there exists a finite $t_1>t_0$ such that $u_i(t_1)<\dot{\gamma}_i^{\mathrm{upper}}$, which contradicts the assumption of persistent upper saturation. Hence, upper saturation cannot be maintained indefinitely.

\textit{Case B (Lower Saturation):} 
Assume instead that $\dot{\gamma}_i(t)=\dot{\gamma}_i^{\mathrm{lower}}$ for all $t\geqslant t_0$. Then, $\dot{e}_{i_1}=\omega_j-\dot{\gamma}_i^{lower}>0$, hence $e_{i_1}(t)$ tends to $+\infty$ as $t$ tends to $\infty$. By the same argument as above, the desired control $u_i(t)$ diverges to $+\infty$, and hence there exists a finite time $t_2>t_0$ such that $u_i(t_2)>\dot{\gamma}_j^{\mathrm{lower}}$, 
showing that sustained lower saturation is also impossible.

\begin{remark}
    \label{rem:sustained saturation}
    If $\omega_j=\dot{\gamma}_i^{upper}$ or $\omega_j=\dot{\gamma}_i^{lower}$, then $\dot{e}_{i_1}=0$ during saturation and error remains constant. This equilibrium corresponds to a non-generic, measure-zero case and is unstable to arbitrarily small perturbations. The strict condition $\dot{\gamma}_i^{lower}<\omega_j<\dot{\gamma}_i^{upper}$ ensures that $\dot{e}_{i_1}(t) \neq 0$, guaranteeing finite-time exit from saturation. 
\end{remark}

Hence, under admissible asymmetric actuator limits, any saturation of the follower’s angular speed is necessarily transient, and the desired control re-enters the admissible interval in finite time. Hence, proved.
\end{proof}

With the confirmation that the desired control input always returns to the admissible range, we can state the next result.
\begin{cor}
    \label{lem:dubin_convergence_2}
    In a system of $n$ pursuers, if any pursuer $j\in \mathcal{V}$ moves on a circular path and its in-neighbour $i\in \mathcal{F}$ follows the guidance law from Theorem \ref{thm:two_pursuers} with actuator bounds $\{\dot\gamma_i^{\mathrm{lower}},\dot\gamma_i^{\mathrm{upper}}\}$ such that $\dot\gamma_i^{\mathrm{lower}}<\omega_j<\dot\gamma_i^{\mathrm{upper}}$, then pursuer $i$ is guaranteed to converge to a circular trajectory around the target with the same angular speed as $j$.
\end{cor}
\begin{proof}
    From Theorem \ref{lem:stauration_behaviour}, we know that the desired control input enters the admissible interval $[\dot\gamma_i^{\mathrm{lower}},\dot\gamma_i^{\mathrm{upper}}]$ within a finite time. Once inside these bounds, the guidance law from Theorem \ref{thm:two_pursuers} drives pursuer $i$ towards convergence. While the control input for pursuer $i$ may saturate multiple times during the engagement, each saturation period is finite. The system repeatedly drives the control back into the non-saturated region, progressively moving towards the stable equilibrium. Eventually, pursuer $i$ settles onto its stable circular orbit.
\end{proof}

By combining the preceding results, we can state a more general result for a network of Dubins vehicles.
\begin{thm}
    \label{thm:dubin_general_convergence}
Let $\mathcal{V}_d$ denote a set of $n$ asymmetric Dubins vehicles communicating over a directed graph 
$\mathcal{G}_{d_C}=(\mathcal{V}_d,\mathcal{E}_{d_C})$.  
The node set $\mathcal{V}_d$ is partitioned into leaders $\mathcal{L}_d\subset\mathcal{V}_d$ and followers 
$\mathcal{F}_d=\mathcal{V}_d\setminus\mathcal{L}_d$. For each vehicle $i\in\mathcal{V}^d$, the state $(x_i(t),y_i(t),\gamma_i(t))$ evolves according to the Dubins kinematics given in eqn. \eqref{eq:kinematics_dubin}, where $v_i$ remains constant and the applied control satisfies $u_i(t)\in[\dot{\gamma}_i^{\mathrm{lower}},\,\dot{\gamma}_i^{\mathrm{upper}}]$ with
$\dot{\gamma}_i^{\mathrm{lower}}<0<\dot{\gamma}_i^{\mathrm{upper}}$. Assume the following conditions hold:
\begin{enumerate}
    \item[(a)] Each leader $l\in\mathcal{L}_d$ circumnavigates the target at a feasible desired circular orbit as guaranteed by 
    Lemma \ref{lem:feas_traj}, i.e., $\lim_{t\to t_{l_f}} r_l(t)=\bar{r}_l$ and $\lim_{t\to t_{l_f}} \dot{\gamma}_l(t)=\bar{\omega}_l$ after some finite time $t_{l_f}$.
    \item[(b)] Each follower $f\in\mathcal{F}_d$, having an out-neighbour $j\in \mathcal{V}_d$, applies the
    saturated guidance law
    \begin{align}
    \label{ew:satu_gidance}
        u_f(t) = \mathrm{sat}\big(g(e_{f_1},e_{f_2}\big),
    \end{align}
    where $\mathrm{sat}(\cdot)$ projects the desired angular rate into $[\dot{\gamma}_f^{\mathrm{lower}},\,\dot{\gamma}_f^{\mathrm{upper}}]$ and function $g(\cdot)$ is defined similarly as in eqn. \eqref{eq:follower_guidance}. The actuator bounds satisfy $\dot{\gamma}_f^{\mathrm{lower}} < \bar{\omega}_{j} < \dot{\gamma}_f^{\mathrm{upper}}$, ensuring the feasibility of tracking the out-neighbour’s steady-state orbit.

    \item[(c)] For every follower $f\in\mathcal{F}_d$, there exists at least one directed path in $\mathcal{G}_{d_C}$
    from $f$ to some leader $l\in\mathcal{L}_d$.
\end{enumerate}

Then, for all $i\in\mathcal{V}$, $\lim_{t\to\infty} r_i(t)=\bar{r}_i$ and $\lim_{t\to\infty} \dot{\gamma}_i(t) = \bar{\omega}_i$, where $\bar{\omega}_i$ equals the steady-state angular speed $\bar{\omega}_{l_m}$ of the leader $l_m$ that terminates the directed path containing $i$.
\end{thm}
\begin{proof}
    The proof is a direct consequence of combining the results from Theorems \ref{thm:two_pursuers}, \ref{lem:stauration_behaviour} with Corollary \ref{lem:dubin_convergence_2}.
\end{proof}
\begin{note}
Here, we consider input saturation only for followers and not for leaders. But, it is possible that the control input of the leader exceeds its bounds. In such cases, we can increase the radius of curvature of leaders by adjusting the length of the trajectory suitably, using techniques for trajectory elongation presented in \cite{rao2024curvature,rao2025trajectory}. This trajectory adjustment is sufficient to keep control input under bounds.    
\end{note} 

\begin{table}[h]
\caption{Simulation's initial conditions (in SI units).}
\centering
\begin{tabular}{|c|c|c|c|c|c|c|}
\hline
 & \multicolumn{3}{c|}{\text{Case 1}} & \multicolumn{3}{c|}{\text{Case 2}} \\
\hline
Pursuer 
& $P$ & $V$ & $\gamma$ & $P$ & $V$ & $\gamma$ \\
\hline
1 & (2,1)     & 28 & -2.53 & (2,1)    & 28 & -2.53\\
\hline
2 & (34,-5) & 35 & -1.57  & (34,-5)  & 35 & -1.57\\
\hline
3 & (-10,15)  & 50 & 0.52   & (-10,15) & 50 & 0.52\\
\hline
4 & (10,-25)  & 18 & 1.05  & (10,-25) & 18 & 1.05\\
\hline
5 & (-25,-15) & 22 & -1.31  & (-25,-15)& 22 & -1.31\\
\hline
6 & NA         & NA  & NA             & (-31,33)  & 87 & 1.43\\ 
\hline
\end{tabular}
\label{tab1}
\end{table}
With these results, we have established a complete framework for achieving circular formations with sparse target information, even under unequal input saturation constraints in both directions, provided a communication path exists from every follower node to at least one leader node. The following section provides simulation results to validate these findings.
\section{Simulation results}
\label{Sec:Simulaton}
The simulations were conducted with a stationary target positioned at the origin. We analyse two primary configurations: the first scenario involves a five-pursuer system with a single leader, and the second a six-pursuer system with two leaders. Asymmetric input saturation is applied in both cases. In all the simulations result, the control input $g(e_{i_1},e_{i_2})$ for every pursuer $i$ is selected as $ C_1e_{i_1} + C_2 \sin{(e_{i_2})}$. In all graphical depictions, leader nodes are coloured blue, and follower nodes are orange.
    \begin{figure}[H]
    \centering
        \begin{subfigure}[b]{0.2\textwidth}
        \centering    
        \includegraphics[width=0.8\linewidth]{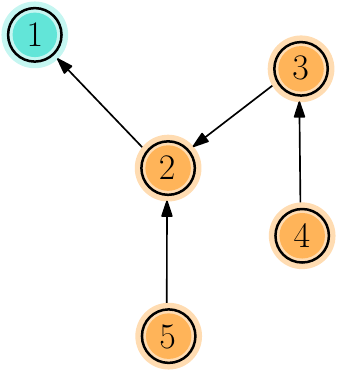}
        \caption{Case $1$}
        \label{fig:graph_one_leader}
    \end{subfigure}
    \begin{subfigure}[b]{0.2\textwidth}
        \centering
        \includegraphics[width=0.9\linewidth]{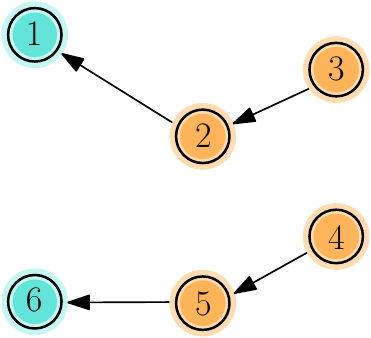}  
        \caption{Case $2$}
        \label{fig:graph_two_leader_second_part}
    \end{subfigure}
    \caption{Communication graphs}
    \label{fig:graph_two_leader_two parts}
\end{figure}
   \begin{figure}[ht]
     \centering
     \begin{subfigure}[b]{0.24\textwidth}
    \centering 
    \includegraphics[width=\linewidth]{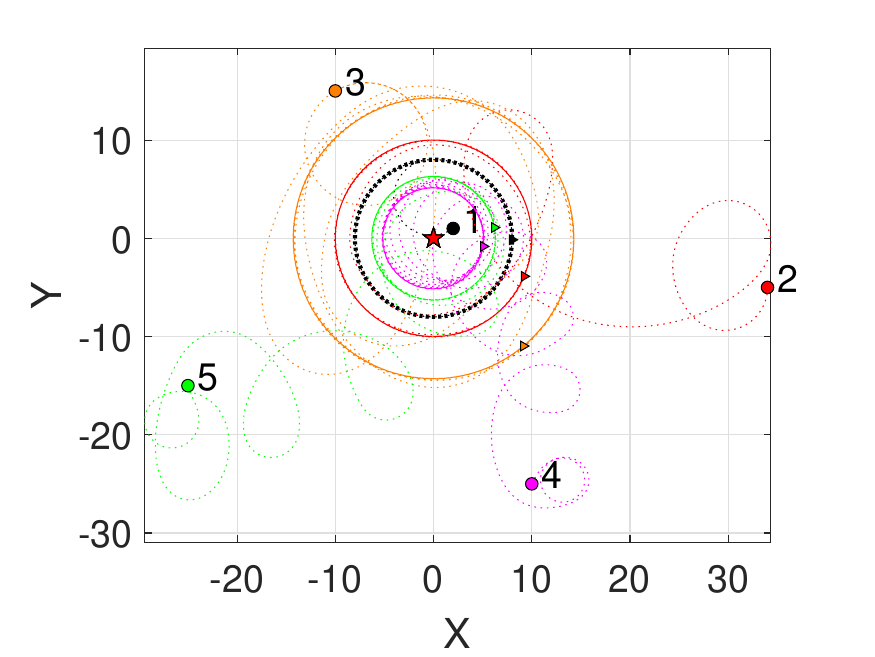}
    \caption{Case $1$}
    \label{fig:trajectory_case3}
     \end{subfigure}
     \begin{subfigure}[b]{0.24\textwidth}
    \centering 
    \includegraphics[width=\linewidth]{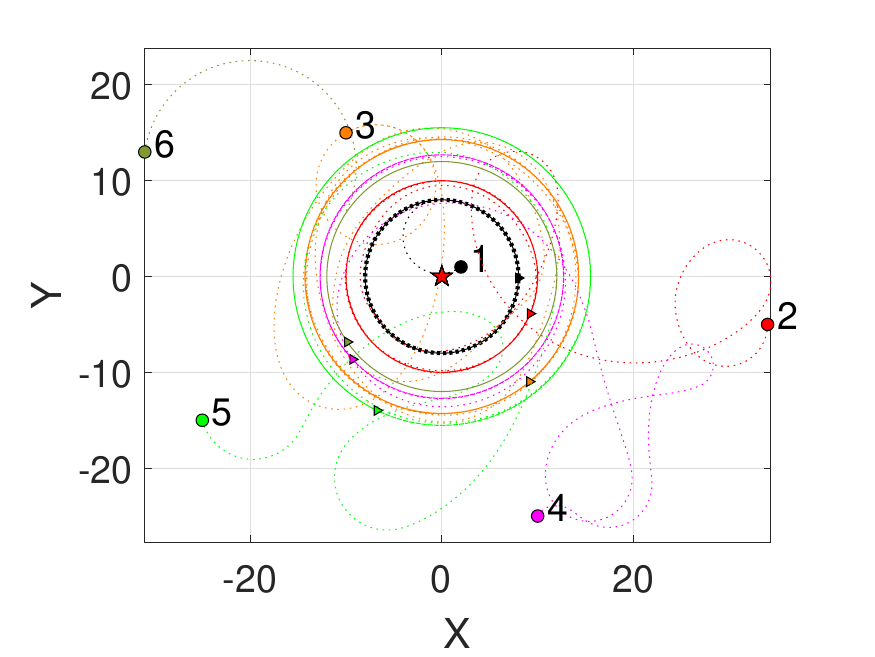}
    \caption{Case $2$}
    \label{fig:trajectory_case4}
     \end{subfigure}
     \caption{Trajectories.}
     \label{fig:trajectories}
\end{figure}
\begin{figure}[ht]
     \centering
     \begin{subfigure}[b]{0.24\textwidth}
    \centering 
    \includegraphics[width=\linewidth]{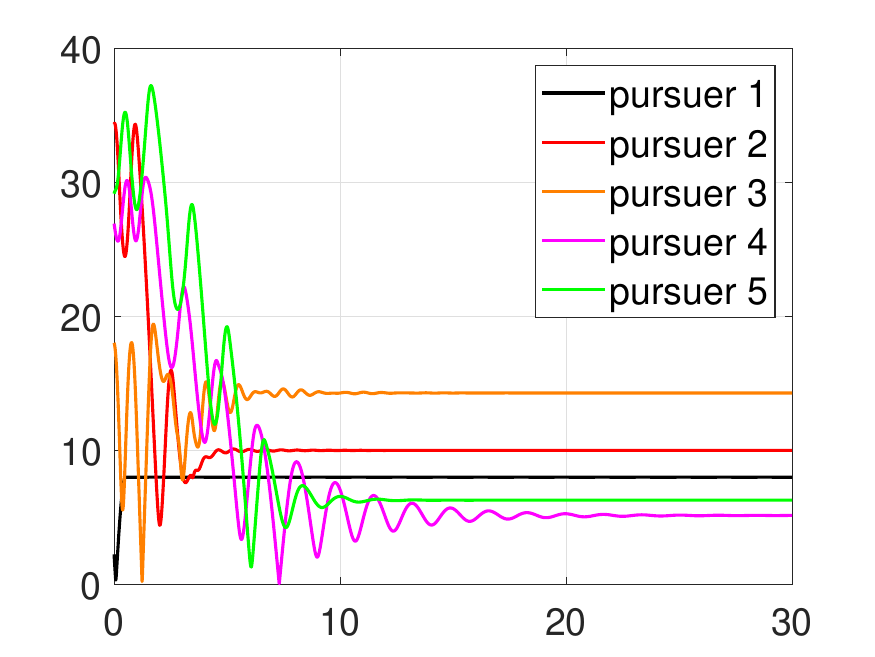}
    \caption{Case $1$}
    \label{fig:distance_from_centre_case3}
     \end{subfigure}
     \begin{subfigure}[b]{0.24\textwidth}
    \centering 
    \includegraphics[width=\linewidth]{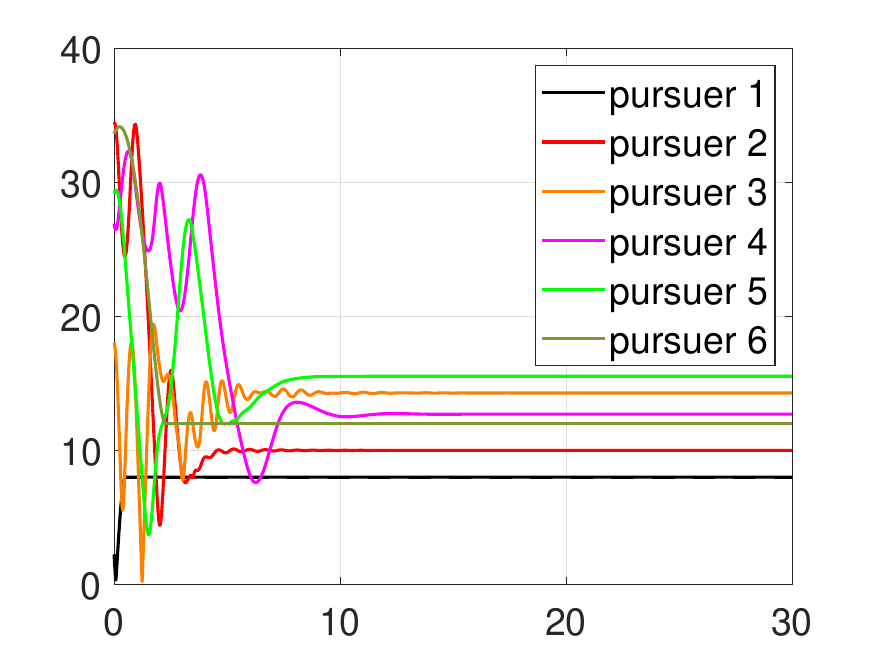}
    \caption{Case $2$}
    \label{fig:distance_from_centre_case4}
     \end{subfigure}
     \caption{Distances from the target}
     \label{fig:distances _from_target}
\end{figure}
\begin{figure}[ht]
     \centering
     \begin{subfigure}[b]{0.24\textwidth}
    \centering 
    \includegraphics[width=\linewidth]{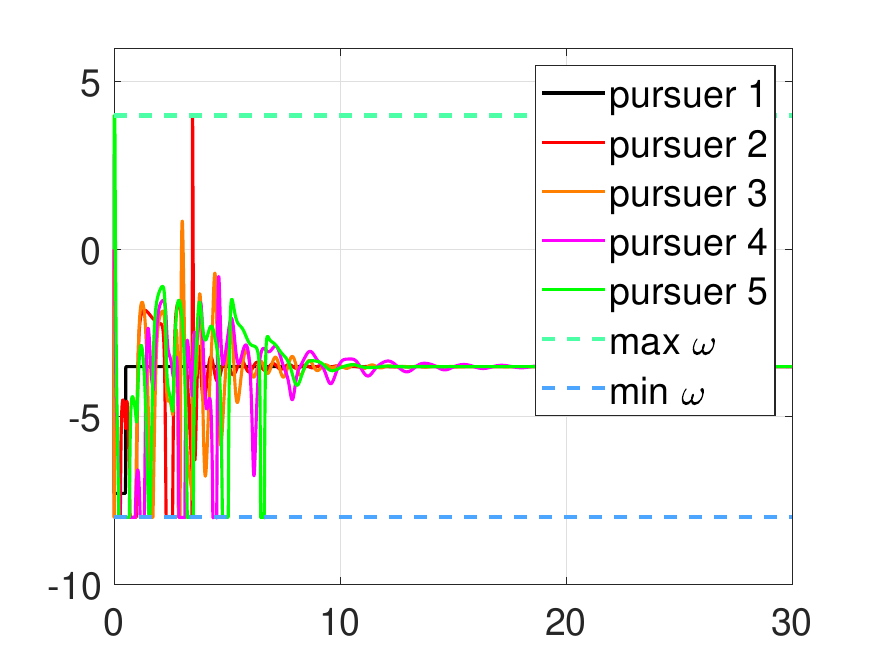}
    \caption{Case $1$}
    \label{fig:control_input_case3}
     \end{subfigure}
     \begin{subfigure}[b]{0.24\textwidth}
    \centering 
    \includegraphics[width=\linewidth]{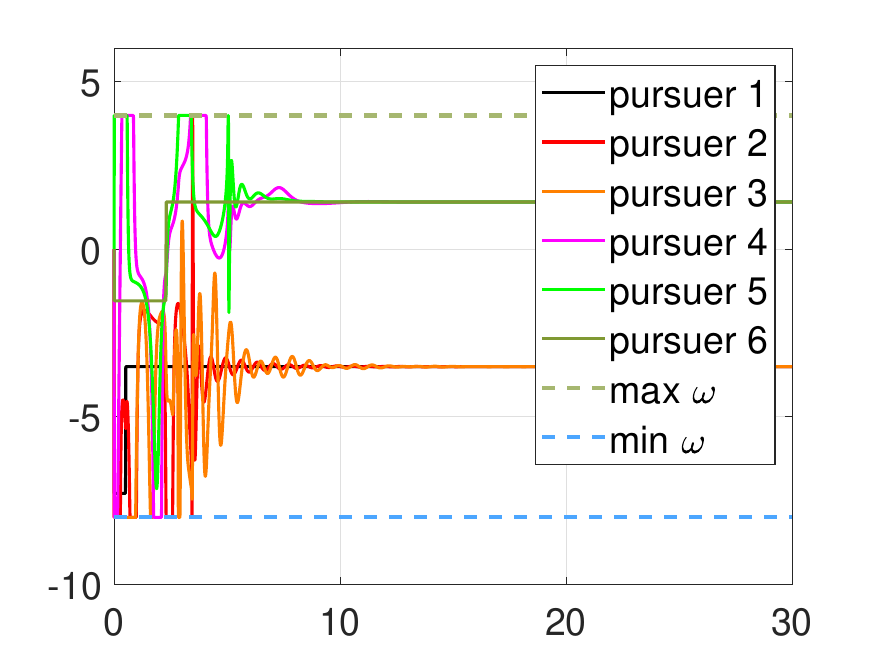}
    \caption{Case $2$}
    \label{fig:control_input_case4}
     \end{subfigure}
     \caption{Control inputs $\omega$}
     \label{fig:control_inputs}
\end{figure}
\subsection*{Case $1$: Single leader with input saturation}
In this scenario, we demonstrate the target's circumnavigation by five pursuers. The communication graph is depicted in Fig. \ref{fig:graph_one_leader}. The initial conditions are stated in Table \ref{tab1}, with control input saturation. $\omega_{max}$ is $4$ rad/s and $\omega_{min}$ is $-8$ rad/s.

As established by Lemmas \ref{lem:stauration_behaviour}, \ref{lem:dubin_convergence_2}, and Theorem \ref{thm:dubin_general_convergence}, the pursuers successfully converge to their designated circular orbits despite input constraints. The trajectories are illustrated in Fig. \ref{fig:trajectory_case3}. Convergence is confirmed by Fig. \ref{fig:distance_from_centre_case3}. The control input plot in Fig. \ref{fig:control_input_case3} shows that all inputs converge to the same value after periods of asymmetric saturation, ultimately forming a single rotating structure.
\subsection*{Case $2$: Multiple leaders with input saturation}
The communication topology for this six pursuer system is shown in Fig.~\ref{fig:graph_two_leader_two parts}, where two leaders are considered ($\alpha=2/3$ as per Proposition~\ref{propo:leaders}). The initial conditions are in Table \ref{tab1}, with asymmetric control input saturation similar to Case $1$.

Consistent with the theoretical framework, the pursuers successfully converge to their circular orbits, despite the presence of multiple leaders and input saturation. The pursuers leave saturation in finite time as predicted in Theorem \ref{lem:stauration_behaviour}. The resulting trajectories are depicted in Fig. \ref{fig:trajectory_case4}, while the stabilised distances from the centre are confirmed in Fig. \ref{fig:distance_from_centre_case4}. Fig. \ref{fig:control_input_case4} illustrates that, although the control inputs are saturated at $-8$ rad/s and $4$ rad/s, they eventually settle at two distinct angular speeds: $\omega_2$ and $\omega_3$ synchronise with $\omega_1$, $\omega_4$ and $\omega_5$ synchronise with $\omega_6$. This confirms the presence of two leaders and leads to the formation of two rotating formations.
\section{Conclusion}
\label{Sec:Conclusion}
In this paper, we present a solution to the problem of distributed circumnavigation of a stationary target using a heterogeneous group of asymmetric Dubins vehicles. The distributed system consists of leaders, who know the target's location, and followers, who do not. The guidance is designed using each pursuer's angular speed only, while their linear speeds are assumed to be constant, yet distinct. The leaders use trajectories formed by concatenating two circular arcs as established in Lemma~\ref{lem:feas_traj}. When there are multiple leaders, the trajectory construction additionally relies on Proposition~\ref{propo:leaders} in conjunction with Lemma~\ref{lem:feas_traj}.
Thereafter, a class of guidance laws are proposed for the followers in Theorem \ref{thm:two_pursuers}. The key advantage of the proposed guidance law is the requirement of only angular information of the heading angle and the LOS angles of a follower's out-neighbour. The framework is then extended to asymmetric Dubins vehicles. The pertinent results presented in Theorem \ref{lem:stauration_behaviour}, Corollary \ref{lem:dubin_convergence_2}, and Theorem \ref{thm:dubin_general_convergence} guarantee the convergence of all Dubins vehicles to their respective circular orbits even under asymmetric input saturation. The results are validated through numerical simulations. As future research, this methodology can be extended to scenarios with a moving target, obstacles, and time-varying communication topologies.

\section*{Acknowledgements}
The authors would like to express their sincere gratitude to Mr. Aditya Krishna Rao for his insightful feedback, which helped enhance the quality of this paper.

\bibliographystyle{ieeetran}        
\bibliography{autosam}


 





\end{document}